\RequirePackage{fix-cm}
\documentclass{svjour3}                     

\smartqed  

\usepackage{graphicx}
\usepackage{amsmath}
\usepackage{amsfonts}
\usepackage{amssymb}
\usepackage{amsopn}
\usepackage{theorem}
\usepackage{latexsym}
\usepackage{graphicx}
\usepackage{mathrsfs}
\usepackage{psfrag}
\usepackage{color}
\usepackage{verbatim}
\usepackage{mathtools}
\usepackage{mathrsfs}
\usepackage{url}
\usepackage{hyperref}

\hypersetup{
 hidelinks,
hypertexnames=false,
}

\usepackage[braket]{qcircuit}
\usepackage{booktabs}

\DeclarePairedDelimiter{\floor}{\lfloor}{\rfloor}
\DeclarePairedDelimiter{\norm}{\lVert}{\rVert}
\DeclarePairedDelimiter{\abs}{\lvert}{\rvert}

\begin{document}

\title{A Quantum Approach to the Discretizable Molecular Distance Geometry Problem
}


\author{Carlile Lavor         \and
        Franklin Marquezino   \and
        Andr\^es Oliveira     \and
        Renato Portugal
}


\institute{C. Lavor \at
              University of Campinas, IMECC, Unicamp, 13081-970, Campinas, SP, Brazil\\
              \email{clavor@ime.unicamp.br}           
           \and
		F. Marquezino \at 
		Federal University of Rio de Janeiro,
		21941-598, Rio de Janeiro, RJ, Brazil\\
		\email{franklin@cos.ufrj.br}
		\and
        A. Oliveira \at 
		University of Campinas, IMECC, UNICAMP,
		13081-970, Campinas, SP, Brazil \\
		\email{andresroliveira98@gmail.com}
        \and
        R. Portugal \at
        National Laboratory of Scientific Computing (LNCC),
		25651-075, Petr\'{o}polis, RJ, Brazil\\
		\email{portugal@lncc.br}
}

\date{Received: date / Accepted: date}

\maketitle


\begin{abstract}
The Discretizable Molecular Distance Geometry Problem (DMDGP) aims to determine the three-dimensional protein structure using distance information from nuclear magnetic resonance experiments. The DMDGP has a finite number of candidate solutions and can be solved by combinatorial methods. We describe a quantum approach to the DMDGP by using Grover's algorithm with an appropriate oracle function, which is more efficient than classical methods that use brute force. We show computational results by implementing our scheme on IBM quantum computers with a small number of noisy qubits.
\end{abstract}

\section{Introduction}
The calculation of the 3D structure of a molecule is a fundamental problem for understanding the molecule function, which is particularly true for proteins~\cite{donald11}. X-ray crystallography was the first method applied to this problem, considering crystallized proteins. For cases where the crystallization is not possible, there is another technique, called Nuclear Magnetic Resonance (NMR)~\cite{wuthrich89}, where (short) Euclidean distances between atoms in a protein are measured. The 3D protein structure determination using this partial distance information~\cite{crippen88} can be modelled as the DMDGP, which has been addressed from a combinatorial approach in~\cite{billinge16,billinge18,lavor17,liberti14a,liberti16,liberti17,mucherino13}.

Quantum computing promises to speed up many important computational tasks, such as searching unsorted databases~\cite{grover97}. Recently, the quantum supremacy has been established by Google~\cite{Arute_2019} and Chinese experiments~\cite{Yulin_Wu_2021}, which means that programmable quantum computers can already execute tasks that classical computers cannot reproduce in a feasible amount of time. The main problem to deliver practical results is the accumulation of noise during the computation and the term NISQ computers~\cite{Pre18} has been used to classify the quantum computers that we are going to use in the near future. 

An attempt to apply quantum computing to the DMDGP was presented in~\cite{lavor05} by exploring the fact that the DMDGP has a finite number of candidate solutions. Grover's algorithm is used to find the solution among the candidates by exhaustive search. 

In this paper, we present a new way of applying Grover's algorithm to the DMDGP by improving the definition of the oracle function. We implement small instances of our algorithm on IBM Quantum with the goal of showing that the procedure works and can be used in practical applications when quantum computers with reasonable size and small error rates are available.

The structure of the paper is as follows. Section~2 defines the DMDGP and presents the combinatorial approach to this problem. Section~3 describes how Grover's algorithm can be used in the DMDGP, which is the main contribution of this paper. Section~4 presents the computational experiments that were implemented on IBM Quantum. Section~5 provides the conclusions and final comments.

\section{The Discretizable Molecular Distance Geometry Problem}

To model our problem, we use a weighted simple undirected graph\footnote{The origin of the word graph is related to representation of molecules and
this relationship is probably the deepest existing between chemistry and
discrete mathematics~\cite{sylvester77}.} $G=(V,E,d)$, where $V$ represents the set of atoms and $E$ represents the set of atom pairs for which a distance is available, given by the function $d:E\rightarrow \lbrack 0,\infty )$.

The problem is then to find a function $x:V\rightarrow \mathbb{R}^{3}$ that
associates each element of $V$ with a point in $\mathbb{R}^{3}$ in such a
way that the Euclidean distances between the points correspond to the values
given by $d$. This is called a Distance Geometry Problem (DGP) in $\mathbb{%
R}^{3}$~\cite{billinge16,billinge18,lavor17,liberti14a,liberti16,liberti17,mucherino13}, formally given as follows:
\begin{problem} (DGP)
Given a simple undirected graph $G=(V,E,d)$ whose edges are weighted by a
function $d:E\rightarrow \lbrack 0,\infty )$, find a function $x\colon
V\rightarrow \mathbb{R}^{3}$ such that%
\begin{equation}
\norm*{x_{u}-x_{v}}=d_{uv},\,\,\,\forall \{u,v\}\in E,  \label{DGP}
\end{equation}%
where $x_{u}=x(u)$, $x_{v}=x(v)$, $d_{uv}=d\left(\left\{u,v\right\}\right)$, and $\norm*{x_{u}-x_{v}}$
is the Euclidean distance between $x_{u}$ and $x_{v}$.
\end{problem}

There is evidence that a closed-form solution for solving (\ref{DGP}) is not
possible~\cite{bajaj88} and a common approach is to formulate the DGP as a
nonlinear global minimization problem~\cite{lara14}, 
\begin{equation*}
\underset{x_{1},\dots,x_{n}\in \mathbb{R}^{3}}{\min }\sum_{\{u,v\}\in
E}\left(\norm*{x_{u}-x_{v}}^{2}-d_{uv}^{2}\right)^{2},
\end{equation*}%
where $|V|=n$. In~\cite{lavor06}, some global optimization algorithms have
been tested but none of them scale well to medium or large instances. A
survey on different methods to the DGP is given in~\cite{liberti10}.

The information provided by NMR experiments and geometric properties of
proteins allow us to define a DGP class, called the \textit{Discretizable
Molecular Distance Geometry Problem }(DMDGP)~\cite{lavor12b}, defined below,
where a combinatorial approach can be applied to the problem~\cite{cassioli15,lavor12c,malliavin19}.
\begin{problem} (DMDGP)
Given a DGP graph $G=(V,E,d)$ and a vertex order $v_{1},\dots,v_{n}$ such that
\begin{itemize}
\item[1.] $v_{1},v_{2},v_{3}$ can be fixed in $\mathbb{R}^{3}$ satisfying (\ref%
{DGP});
\item[2.] $\forall i>3$, the set $\left\{v_{i-3},v_{i-2},v_{i-1},v_{i}\right\}$ is a clique
with%
\begin{equation*}
d_{i-3,i-2}+d_{i-2,i-1}>d_{i-3,i-1},
\end{equation*}
\end{itemize}
find a function $x\colon V\rightarrow \mathbb{R}^{3}$ such that
\begin{equation*}
\norm*{x_{u}-x_{v}}=d_{uv},\,\,\,\forall \{u,v\}\in E.
\end{equation*}
\end{problem}

We list some important observations about this problem: (1)~the DMDGP vertex order is the main idea behind the discrete version of the
DGP (for instance, see~\cite{cassioli13,lavor12a,lavor19a,lavor19b}); (2)~the positions for $v_{1},v_{2},v_{3}$ guarantee that the solution set will contain just incongruent solutions (aside from a single reflection) and the strictness of the triangular inequality prevents an uncountable quantity of solutions~\cite{lavor12b};
(3)~an exact solution method, called
Branch-and-Prune (BP), was presented in~\cite{carvalho08,liberti08} for finding all incongruent solutions (that is, solutions obtained modulo rotations and translations).
The BP algorithm can be exponential in the worst case, which is consistent
with the fact that the DMDGP is an NP-hard problem~\cite{lavor12b}.

The DMDGP order \textquotedblleft organizes\textquotedblright\ the search
space in a binary tree and the additional distance information (related to
the pairs $(v_{j},v_{i})$, $j<i-3$, \thinspace $i=5,\dots,n$) can be used to
reduce the search space by pruning infeasible positions in the tree, which
begins with the fixed positions for $v_{1},v_{2},v_{3}$. The search ends
when a path from the root of the tree to a leaf node is found by the BP
algorithm, such that the positions relative to vertices in the path satisfy
the DGP equations (\ref{DGP}).

There are symmetries\textit{\ }in the BP tree~\cite{gonalves21,lavor21,liberti14b,mucherino12},
related to the cardinality of the DMDGP solution set, which can be computed
even before applying BP algorithm. These symmetric properties are based on
the set

\begin{equation*}
S=\{v\in V:\nexists \{u,w\}\in E\;\text{such that }u+3<v\leq w\},
\end{equation*}%
for a given DMDGP instance $G=(V,E,d)$~\cite{liberti14b} (we denote by $u+3$
the third vertex after $u$). 

To illustrate the importance of these symmetries, let us consider a small DMDGP
instance given by%
\begin{eqnarray*}
V &=&\{v_{1},v_{2},v_{3},v_{4},v_{5},v_{6},v_{7}\}, \\
E &=&\{\{v_{1},v_{2}\},\{v_{1},v_{3}\},\{v_{1},v_{4}\},\{v_{1},v_{6}\}, \\
&&\{v_{2},v_{3}\},\{v_{2},v_{4}\},\{v_{2},v_{5}\}, 
\{v_{3},v_{4}\},\{v_{3},v_{5}\},\{v_{3},v_{6}\}, \\
&&\{v_{4},v_{5}\},\{v_{4},v_{6}\},\{v_{4},v_{7}\}, 
\{v_{5},v_{6}\},\{v_{5},v_{7}\},\{v_{6},v_{7}\}\}.
\end{eqnarray*}%
It is easy to check that, for this example,%
\begin{equation*}
S=\{v_{4},v_{7}\}.
\end{equation*}

To simplify the notation, let us represent the first solution found by BP by
a sequence of zeros and ones (remember that the positions for $%
v_{1},v_{2},v_{3}$ are fixed): 
\begin{equation*}
s_{1}=(0,1,0,1).
\end{equation*}%
The fact that $v_{7}\in S$ implies that another solution is given by 
\begin{equation*}
s_{2}=(0,1,0,0),
\end{equation*}%
and from the fact that $v_{4}\in S$, more two solutions are given by%
\begin{equation*}
s_{3}=(1,0,1,0)
\end{equation*}%
and 
\begin{equation*}
s_{4}=(1,0,1,1).
\end{equation*}

In addition to be possible to generate all the other solutions from just
one, the set $S$ also informs us that the cardinality of the solution set is
given by $2^{|S|}$~\cite{liberti13a}. Thus, we can use Grover's algorithm to find just one solution.

\section{Solving the DMDGP by Grover's algorithm}
This section briefly describes the Grover's algorithm and shows how it can
be used to solve the DMDGP.

\subsection{Grover's algorithm}
Grover's algorithm~\cite{grover97} searches a list of elements labeled by integers from $0$ to $N-1$, where $N=2^{n}$ for some integer $n\geq 2$, for a particular element $i_0$. The algorithm uses two registers whose number of qubits are $n$ and $1$, respectively. The first step creates a superposition of all $2^{n}$ computational basis states $%
\left\{\left\vert 0\right\rangle ,...,\left\vert 2^{n}-1\right\rangle\right\}$ in the
first register, which is achieved by applying the Hadamard operator $H$, 
\begin{equation*}
H=\frac{1}{\sqrt{2}}\left[ 
\begin{array}{rr}
1 & 1 \\ 
1 & -1%
\end{array}%
\right] ,
\end{equation*}%
on each qubit in state $\left\vert 0\right\rangle $,
resulting in
\begin{equation}\label{eq:psi}
|\psi \rangle =\frac{1}{\sqrt{2^{n}}}\sum_{i=0}^{2^{n}-1}\left\vert
i\right\rangle.
\end{equation}
The second register is initialized in state $\left\vert 1\right\rangle $
and, after applying again the Hadamard operator, it changes to the state%
\begin{equation*}
\left\vert -\right\rangle =H|1\rangle =\frac{|0\rangle -|1\rangle }{\sqrt{2}}.
\end{equation*}

Grover's algorithm uses two unitary operators. The first is 
\begin{equation}
U_{f} \left\vert i\right\rangle \left\vert j\right\rangle 
=\left\vert i\right\rangle \left\vert j\oplus f(i)\right\rangle ,
\end{equation}%
where $i\in $ $\{0,\dots,N-1\}$, $j\in $ $\{0,1\}$, $\oplus $ is the sum
modulo 2, and $f:\{0,\dots,N-1\}\rightarrow \{0,1\}$ is a function, called oracle, that
recognizes the searched element $i_{0}$ ($f(i)=1$ if and only if $i=i_{0}$).
Using that $1\oplus f(i_{0})=0$ and $1\oplus f(i)=1$, for all $i\neq
i_{0}$, we obtain%
\begin{equation}
U_{f} \left\vert i\right\rangle \left\vert -\right\rangle 
=(-1)^{f(i)}\left\vert i\right\rangle \left\vert -\right\rangle ,  \label{f}
\end{equation}%
implying that the state of the first register $\left\vert \psi _{1}\right\rangle$ after applying $U_f$ is
\begin{eqnarray*}
\left\vert \psi _{1}\right\rangle   &=&
\frac{1}{\sqrt{N}}\sum_{i=0}^{N-1}(-1)^{f(i)}|i\rangle .
\end{eqnarray*}
The state of the searched element $i_{0}$ is different from the remaining ones because it is the only one with negative amplitude, but
this piece of information is not useful---the probability of finding $i_0$ after a measurement is equal to any other wrong element.

The second unitary operator increases the amplitude of $i_{0}$ and at the same time uniformly decreases the amplitude of the other elements. Its definition is
\begin{equation*}
G= 2\left\vert \psi \right\rangle \left\langle \psi \right\vert
-I,
\end{equation*}%
where $|\psi \rangle$ is given by $(\ref{eq:psi})$ and $I$ is the identity operator. By applying $G$ on the incumbent state $\left\vert \psi _{1}\right\rangle $, we obtain
\begin{eqnarray*}
\left\vert \psi _{G}\right\rangle &=& \sum_{\substack{ i=0\\ i\neq i_{0}}}^{N-1}\frac{N-4}{N\sqrt{N}}|i\rangle 
+\frac{3N-4}{N\sqrt{N}}|i_{0}\rangle .
\end{eqnarray*}

The composition of the two operators is the evolution operator $U=(G\otimes I_2) U_f$.
The action of $U^{k}$ $(k\in \mathbb{N})$ rotates $|\psi \rangle|- \rangle $ towards $|i_{0}\rangle|- \rangle $ by $k\theta $
degrees, in the subspace spanned by $|\psi \rangle|- \rangle $ and $|i_{0}\rangle|- \rangle $,
where $\theta $ is the angle between $|\psi \rangle $ and $U|\psi \rangle $~\cite{lavor07}.
The number of times $k$ that $U$ must be applied so
that the angle between $|i_{0}\rangle|- \rangle $ and $U^{k}|\psi \rangle |- \rangle$ becomes
zero is given by 
\begin{equation*}
k=\frac{\arccos \left( \frac{1}{\sqrt{N}}\right) }{\arccos \left( \frac{N-2}{%
N}\right) }.
\end{equation*}%
When $N$ is large, we have 
\begin{equation*}
\lim_{N\rightarrow \infty }\frac{k}{\sqrt{N}}=\frac{\pi }{4}.
\end{equation*}
Thus, applying $U$ exactly $\lfloor k\rfloor $ times and
measuring the first register, we obtain $i_{0}$ with probability $O(1)$
in $O(\sqrt{N})$ steps~\cite{marquezino19,portugal18}. Grover's algorithm can be extended to query databases with repeated elements~\cite{nielsen00}.

When we apply Grover's algorithm to practical problems, we have to deal with the problem of defining an efficient oracle function. It is not a trivial task to find an efficient oracle function for the DMDGP.

\subsection{The oracle function for the DMDGP}
From the DMDGP definition, all distance values between vertices in the set $%
\{v_{i-3},v_{i-2},v_{i-1},v_{i}\}$, for $i=4,...,n\,$, allow us to acquire the following pieces of information:
\begin{itemize}
\item[1.] $d_{1,2},\ldots ,d_{n-1,n}$ (distances associated to consecutive vertices);

\item[2.] $\theta _{1,3},\ldots ,\theta _{n-2,n}$ (angles in $(0,\pi )$ defined by three consecutive vertices); and

\item[3.] $\cos (\omega _{1,4}),\ldots ,\cos (\omega _{n-3,n})$ (cosines of torsion angles in $[0,2\pi ]$ defined by four consecutive vertices~\cite{lavor15}), given by
{\scriptsize
\begin{equation*}
\cos (\omega _{i-3,i})=\frac{%
2d_{i-2,i-1}^{2}(d_{i-3,i-2}^{2}+d_{i-2,i}^{2}-d_{i-3,i}^{2})-(d_{i-3,i-2,i-1})(d_{i-2,i-1,i})%
}{\sqrt{4d_{i-3,i-2}^{2}d_{i-2,i-1}^{2}-(d_{i-3,i-2,i-1}^{2})}\sqrt{%
4d_{i-2,i-1}^{2}d_{i-2,i}^{2}-(d_{i-2,i-1,i}^{2})}},  
\end{equation*}}
where 
{\scriptsize
\begin{eqnarray*}
d_{i-3,i-2,i-1} &=&d_{i-3,i-2}^{2}+d_{i-2,i-1}^{2}-d_{i-3,i-1}^{2} \\
d_{i-2,i-1,i} &=&d_{i-2,i-1}^{2}+d_{i-2,i}^{2}-d_{i-1,i}^{2}.
\end{eqnarray*}}
\end{itemize}

From $\cos (\omega _{i-3,i})$, for $i=4,...,n$, we obtain two possible
values for each torsion angle $\omega _{i-3,i}$; and considering that the
DMDGP order $v_{1},\dots,v_{n}$ represents bonded atoms of a molecule, the
values $d_{i-1,i}$, $\theta _{i-2,i}$ and $\omega _{i-3,i}$ are exactly the 
\textit{internal coordinates} of the molecule that can also be used to describe its
structure~\cite{lavor12b}. This means that the 3D molecular structure is
defined by choosing $+$ or $-$ from $\sin (\omega _{i-3,i})=\pm \sqrt{1-\cos
^{2}(\omega _{i-3,i})}$, for $i=4,\ldots ,n$ (the signs $+$ and $-$ are
associated to the branches of the BP tree).

The atomic Cartesian coordinates $x_{i}=(x_{i_{1}},x_{i_{2}},x_{i_{3}})^{T}%
\in \mathbb{R}^{3}$ can be obtained from the internal coordinates using the
following matrix multiplications~\cite{lavor12b}:

\begin{equation}
\left[ 
\begin{array}{r}
x_{i_{1}} \\ 
x_{i_{2}} \\ 
x_{i_{3}} \\ 
1%
\end{array}%
\right] =B_{1}B_{2}\cdots B_{i}\left[ 
\begin{array}{r}
0 \\ 
0 \\ 
0 \\ 
1%
\end{array}%
\right] ,\text{ }\forall i=1,\ldots ,n,  \label{xi}
\end{equation}%
where 
\begin{equation*}
B_{1}=\left[ 
\begin{array}{rrrr}
1 & 0 & 0 & 0 \\ 
0 & 1 & 0 & 0 \\ 
0 & 0 & 1 & 0 \\ 
0 & 0 & 0 & 1%
\end{array}%
\right] ,\text{ }B_{2}=\left[ 
\begin{array}{rrrr}
-1 & 0 & 0 & -d_{1,2} \\ 
0 & 1 & 0 & 0 \\ 
0 & 0 & -1 & 0 \\ 
0 & 0 & 0 & 1%
\end{array}%
\right] ,
\end{equation*}%
\begin{equation*}
B_{3}=\left[ 
\begin{array}{rrrr}
-\cos \theta _{1,3} & -\sin \theta _{1,3} & 0 & -d_{2,3}\cos \theta _{1,3}
\\ 
\sin \theta _{1,3} & -\cos \theta _{1,3} & 0 & d_{2,3}\sin \theta _{1,3} \\ 
0 & 0 & 1 & 0 \\ 
0 & 0 & 0 & 1%
\end{array}%
\right] ,
\end{equation*}%
and 
\begin{equation*}
B_{i}=\left[ 
\begin{array}{rrrr}
-\cos \theta _{i-2,i} & -\sin \theta _{i-2,i} & 0 & -d_{i-1,i}\cos \theta
_{i-2,i} \\ 
\sin \theta _{i-2,i}\cos \omega _{i-3,i} & -\cos \theta _{i-2,i}\cos \omega
_{i-3,i} & -\sin \omega _{i-3,i} & d_{i-1,i}\sin \theta _{i-2,i}\cos \omega
_{i-3,i} \\ 
\sin \theta _{i-2,i}\sin \omega _{i-3,i} & -\cos \theta _{i-2,i}\sin \omega
_{i-3,i} & \cos \omega _{i-3,i} & d_{i-1,i}\sin \theta _{i-2,i}\sin \omega
_{i-3,i} \\ 
0 & 0 & 0 & 1%
\end{array}%
\right] ,
\end{equation*}%
for $i=4,...,n$.

From these matrices, the first three atoms of the molecule can be fixed at
positions%
\begin{equation*}
x_{1}=\left[ 
\begin{array}{r}
0 \\ 
0 \\ 
0%
\end{array}%
\right] ,\text{ }x_{2}=\left[ 
\begin{array}{r}
-d_{1,2} \\ 
0 \\ 
0%
\end{array}%
\right] ,\text{ }x_{3}=\left[ 
\begin{array}{r}
-d_{1,2}+d_{2,3}\cos \theta _{1,3} \\ 
d_{2,3}\sin \theta _{1,3} \\ 
0%
\end{array}%
\right] ,
\end{equation*}%
implying that there are $N=2^{n-3}$ possible configurations for a molecule
with $n$ atoms related to the states $\left\vert 0\right\rangle
,...,\left\vert N-1\right\rangle $ of the first register of Grover's
algorithm.

For each qubit of the first register of Grover's algorithm, states $%
\left\vert 0\right\rangle $ and $\left\vert 1\right\rangle $ are associated
to $+\sqrt{1-\cos ^{2}(\omega _{i-3,i})}$ and $-\sqrt{1-\cos ^{2}(\omega
_{i-3,i})}$, respectively, for $i=4,...,n$.

For $n=6$, the possible states are 
\begin{eqnarray*}
\left\vert 0\right\rangle  &=&\left\vert 000\right\rangle \text{, }%
\left\vert 1\right\rangle =\left\vert 001\right\rangle \text{, }\left\vert
2\right\rangle =\left\vert 010\right\rangle \text{, }\left\vert
3\right\rangle =\left\vert 011\right\rangle \text{, } \\
\left\vert 4\right\rangle  &=&\left\vert 100\right\rangle \text{, }%
\left\vert 5\right\rangle =\left\vert 101\right\rangle \text{, }\left\vert
6\right\rangle =\left\vert 110\right\rangle \text{, }\left\vert
7\right\rangle =\left\vert 111\right\rangle \text{,}
\end{eqnarray*}%
where $\left\vert 5\right\rangle =\left\vert 101\right\rangle $, for
example, is represented by

\begin{equation*}
-\sqrt{1-\cos ^{2}(\omega _{1,4})}\text{, }+\sqrt{1-\cos ^{2}(\omega _{2,5})}%
\text{, }-\sqrt{1-\cos ^{2}(\omega _{3,6})}\text{.}
\end{equation*}

Given a candidate solution $\left\vert k\right\rangle $, $k=0,\dots,2^{n-3}-1$%
, we define a function%
\begin{equation*}
h:\{0,\dots,2^{n-3}-1\}\rightarrow \mathbb{R}^{3n}
\end{equation*}%
by%
\begin{equation}
h(k)=\left(x_{1}^{k},\dots,x_{n}^{k}\right),  \label{H}
\end{equation}%
with $x_{i}^{k}$ given by (\ref{xi}), for $i=1,\dots,n$.

For checking if $h(k)$, for $k=0,\dots,2^{n-3}-1$, is a DMDGP solution, we
define another function%
\begin{equation*}
g:\mathbb{R}^{3n}\rightarrow \mathbb{R}
\end{equation*}%
by%
\begin{equation}
g\left(x_{1}^{k},...,x_{n}^{k}\right)=\sum_{\{u,v\}\in
E}\left(\norm*{x_{u}^{k}-x_{v}^{k}}^{2}-d_{uv}^{2}\right)^{2},  \label{G}
\end{equation}%
where $G=(V,E,d)$ is a DMDGP instance. A DMDGP solution is obtained if, and
only if, $g\left(x_{1}^{k},...,x_{n}^{k}\right)=0$.

In order to define an oracle function $f$ for the Grover's algorithm, we
first prove that it is possible to define a parameter $p_{1}$, in terms of $n
$, such that 
\begin{equation*}
\frac{g(h(k))}{p_{1}}\in \lbrack 0,1],
\end{equation*}%
for all $k=0,...,2^{n-3}-1$.

\begin{proposition} \label{p1}
Given the function $g\circ h:\left\{0,\dots,2^{n-3}-1\right\}\rightarrow \mathbb{R}$,
where $h$ and $g$ are defined by (\ref{H}) and (\ref{G}), respectively, we
have that%
\begin{equation*}
\frac{g(h(k))}{6^{4}(n^{6}+n^{2})}\in \lbrack 0,1].
\end{equation*}
\end{proposition}

\begin{proof}
Let us define candidate solutions of a DMDGP instance $G=(V,E,d)$ by $%
x^{k}=\left(x_{1}^{k},...,x_{n}^{k}\right)\in \mathbb{R}^{3n}$, for $k=0,\dots,2^{n-3}-1$%
, which implies that%
\begin{equation*}
g(h(k))=\sum_{\{u,v\}\in E}\left(\norm*{x_{u}^{k}-x_{v}^{k}}^{2}-d_{uv}^{2}\right)^{2}.
\end{equation*}
We want to define $p_{1}$ such that, for $k=0,\dots,2^{n-3}-1$,%
\begin{equation*}
0\leq \sum_{\{u,v\}\in E}\left(\norm*{x_{u}^{k}-x_{v}^{k}}^{2}-d_{uv}^{2}\right)^{2}\leq p_{1}\text{.}
\end{equation*}
Since there are $O(n^{2})$ terms in this sum, if we find a value $s$ such
that%
\begin{equation*}
\underset{k=0,...,2^{n-3}-1}{\underset{\{u,v\}\in E}{\max}}%
\left\{\left(\norm*{x_{u}^{k}-x_{v}^{k}}^{2}-d_{uv}^{2}\right)^{2}\right\} \leq s,
\end{equation*}
we could define $p_{1}=n^{2}s$.

For $k=0,...,2^{n-3}-1$, we have
\begin{align*}
\underset{\{u,v\}\in E}{\max }\left\{\left(\norm*{x_{u}^{k}-x_{v}^{k}}^{2}-d_{uv}^{2}\right)^{2}\right\} &=
\underset{\{u,v\}\in E}{\max}\left\{\norm*{x_{u}^{k}-x_{v}^{k}}^{4}+d_{uv}^{4}-2\norm*{x_{u}^{k}-x_{v}^{k}}^{2}d_{uv}^{2}\right\}
\\
&\leq \underset{\{u,v\}\in E}{\max }\left\{\norm*{x_{u}^{k}-x_{v}^{k}}^{4}\right\}+%
\underset{\{u,v\}\in E}{\max }\left\{d_{uv}^{4}\right\}.
\end{align*}
The maximum value for $d_{uv}$ is related to the maximum value detected by
NMR experiments (given by $6\mathring{A}$)~\cite{donald11}, and the maximum
value for $\norm*{x_{u}^{k}-x_{v}^{k}}$ is obtained when $\theta _{i-2,i}=\pi $,
for $i=3,...,n$, and $\omega _{i-3,i}=0$, for for $i=4,...,n$. Considering
also $6\mathring{A}$ for $d_{i-1,i}$, for $i=2,...,n$ (which is greater than
the maximum value for a covalent bond in proteins~\cite{donald11}), we obtain%
\begin{equation*}
\underset{k=0,...,2^{n-3}-1}{\underset{\{u,v\}\in E}{\max }}%
\left\{\left(\norm*{x_{u}^{k}-x_{v}^{k}}^{2}-d_{uv}^{2}\right)^{2}\right\}\leq
(6n)^{4}+6^{4}=6^{4}(n^{4}+1).
\end{equation*}
Hence, for $s=6^{4}(n^{4}+1)$ and $p_{1}=n^{2}s$, we obtain%
\begin{equation*}
\frac{g(h(k))}{6^{4}(n^{6}+n^{2})}\in \lbrack 0,1].
\end{equation*}
\qed
\end{proof}

Now, we show that it is also possible to define another parameter $p_{2}$
such that, for all $k=0,...,2^{n-3}-1$ and a given $\epsilon > 0$,  
\begin{equation*}
\left( \frac{g(h(k))}{p_{1}}\right) ^{\frac{1}{p_{2}}}\in \lbrack
0,1-\epsilon ),
\end{equation*}%
when $k$ is associated to a DMDGP solution, and%
\begin{equation*}
\left( \frac{g(h(k))}{p_{1}}\right) ^{\frac{1}{p_{2}}}\in \lbrack 1-\epsilon,1],
\end{equation*}%
in the other cases.

From (\ref{G}), a DMDGP solution $x_{1}^{k},\dots,x_{n}^{k}$ is identified
when $g\left(x_{1}^{k},...,x_{n}^{k}\right)=0$. However, due to rounding errors in
fixed-point arithmetics, we define a DMDGP solution $x_{1}^{k},\dots,x_{n}^{k}$
if, and only if,%
\begin{equation*}
g\left(x_{1}^{k},\dots,x_{n}^{k}\right)<\delta ,
\end{equation*}%
for a given small value $\delta >0$.

\begin{theorem}\label{p2}
Given $\delta >0$, the function $g\circ h:\{0,...,2^{n-3}-1\}\rightarrow 
\mathbb{R}$, where $h$ and $g$ are defined by (\ref{H}) and (\ref{G}),
respectively, and a value $\epsilon >0$ such that $\delta <1-\epsilon $, we
have that%
\begin{equation*}
\left( \frac{g(h(k))}{6^{4}(n^{6}+n^{2})}\right) ^{\frac{1}{\log
_{1-\epsilon }\left( \frac{\delta }{6^{4}(n^{6}+n^{2})}\right) }}\in \lbrack
0,1-\epsilon )\text{,}
\end{equation*}
when $k$ is associated to a DMDGP solution ($g(h(k))<\delta $), and  
\begin{equation*}
\left( \frac{g(h(k))}{6^{4}(n^{6}+n^{2})}\right) ^{\frac{1}{\log
_{1-\epsilon }\left( \frac{\delta }{6^{4}(n^{6}+n^{2})}\right) }}\in \lbrack
1-\epsilon ,1],
\end{equation*}
when $k$ is not associated to a DMDGP solution ($g(h(k))\geq \delta$).
\end{theorem}

\begin{proof}
First, note that, for $p_{1}=6^{4}(n^{6}+n^{2})$,%
\begin{align*}
\frac{\delta }{p_{1}} < \delta <1-\epsilon &\Rightarrow  \ln \left( \frac{\delta }{p_{1}}\right)  <0\text{ and }\ln (1-\epsilon )<0 \\
&\Rightarrow \frac{\ln \left( \frac{\delta }{p_{1}}\right) }{\ln (1-\epsilon )} =\log_{1-\epsilon }\left( \frac{\delta }{p_{1}}\right) >0.
\end{align*}
Defining%
\begin{equation*}
p_{2}=\log _{1-\epsilon }\left( \frac{\delta }{p_{1}}\right) 
\end{equation*}
and considering a DMDGP solution $k$ ($g(h(k))<\delta $), we obtain%
\begin{align*}
0 \leq g(h(k))<\delta &\Rightarrow 0\leq \frac{g(h(i_{0}))}{p_{1}}<\frac{ \delta }{p_{1}} \\
&\Rightarrow 0\leq \left( \frac{g(h(k))}{p_{1}}\right) ^{\frac{1}{p_2}}<\left( \frac{\delta }{p_{1}}\right) ^{\frac{1}{p_2}} \\
&\Rightarrow 0\leq \left( \frac{g(h(k))}{p_{1}}\right) ^{\frac{1}{p_2}}<\left( \frac{\delta }{p_{1}}\right) ^{\frac{1}{\log _{1-\epsilon }\left( \frac{\delta}{p_{1}}\right)}} \\
&\Rightarrow 0\leq \left( \frac{g(h(i_{0}))}{p_{1}}\right) ^{\frac{1}{p_2}}<1-\epsilon  \\
&\Rightarrow \left( \frac{g(h(k))}{p_{1}}\right) ^{\frac{1}{p_2}}\in \lbrack 0,1-\epsilon ).
\end{align*}

When $k$ is not a DMDGP solution ($g(h(k))\geq \delta $), we get, from Proposition \ref{p1},%
\begin{align*}
\delta  \leq g(h(k)) &\Rightarrow \frac{\delta }{p_{1}}\leq \frac{g(h(k))}{p_{1}}\leq 1 \\
&\Rightarrow \left( \frac{\delta }{p_{1}}\right) ^{\frac{1}{p_{2}}}\leq \left( \frac{g(h(k))}{p_{1}}\right) ^{\frac{1}{p_{2}}}\leq 1 \\
&\Rightarrow \left( \frac{\delta }{p_{1}}\right) ^{\frac{1}{\log_{1-\epsilon }\left( \frac{\delta }{p_{1}}\right) }}\leq \left( \frac{g(h(k))}{p_{1}}\right) ^{\frac{1}{p_{2}}}\leq 1 \\
&\Rightarrow 1-\epsilon \leq \left( \frac{g(h(k))}{p_{1}}\right) ^{\frac{1}{p_2}}\leq 1 \\
&\Rightarrow \left( \frac{g(h(k))}{p_{1}}\right) ^{\frac{1}{p_2}}\in \lbrack 1-\epsilon ,1].
\end{align*}
\qed
\end{proof}

\begin{corollary}
For a given $\delta >0$ and $\epsilon >0$ such that $\delta+\epsilon <1 $, the function%
\begin{equation*}
f:\left\{0,\dots,2^{n-3}-1\right\}\rightarrow \{0,1\},
\end{equation*}%
defined by%
\begin{equation*}
f(k)=1-\floor*{ \left( \frac{g(h(k))}{p_{1}}\right) ^{\frac{1}{p_{2}}%
}+\epsilon} ,
\end{equation*}%
where 
\begin{equation*}
p_{1}=6^{4}(n^{6}+n^{2})
\end{equation*}%
and%
\begin{equation*}
p_{2}=\log _{1-\epsilon }\left( \frac{\delta }{p_{1}}\right) ,
\end{equation*}%
satisfies the following property:%
\begin{equation*}
f(k)=1\Longleftrightarrow k\text{ is associated to a DMDGP solution (}%
g(h(k))<\delta \text{).}
\end{equation*}
\end{corollary}

\begin{proof}
From Theorem \ref{p2}, when $k\in \left\{0,...,2^{n-3}-1\right\}$ is not a DMDGP solution ($%
g(h(k))\geq \delta $),%
\begin{equation*}
\left( \frac{g(h(k))}{p_{1}}\right) ^{\frac{1}{p_2}}\in \lbrack 1-\epsilon
,1]\Longrightarrow f(k)=1-\floor*{ \left( \frac{g(h(k))}{p_{1}}\right) ^{%
\frac{1}{p_{2}}}+\epsilon} =0.
\end{equation*}
For a DMDGP solution $k$ ($g(h(k))<\delta $),%
\begin{equation*}
\left( \frac{g(h(k))}{p_{1}}\right) ^{\frac{1}{p_2}}\in \lbrack 0,1-\epsilon
)\Longrightarrow f(k)=1-\floor*{\left( \frac{g(h(k))}{p_{1}}\right) ^{\frac{%
1}{p_{2}}}+\epsilon} =1.
\end{equation*}
\qed
\end{proof}

\section{Computational results}

In this section, we describe the implementation of some instances of the oracle function for the DMDGP on IBM quantum computers. Since the quantum computers that are available nowadays have a high error rate, we restrict to 3-qubit oracles to show that our procedure works and can be extended to meaningful applications when error rates are improved or error correcting codes are implemented.

In our experiments, we have used IBM Qiskit to implement Grover's algorithm. We hard coded oracle circuits for objective functions corresponding to each possible instance of the DMDGP for molecules with \textbf{7} atoms. We ran the experiments on IBM Santiago, IBM Lagos, IBM Bogota, and on IBM QASM Simulator.


In Fig.~\ref{fig:santiago-vs-simulation}, we have a summary of the results obtained from the search algorithm running on IBM Santiago---which was the best device for most of our experiments---when the searched element was $\ket{010}$. The horizontal axis depicts the possible outcomes, and the vertical axis depicts the probability of each outcome based on 8196 independent runs performed on the device. We compare the results obtained using the standard set of gates and the improved decomposition of multi-controlled $Z$ gates. The results for an ideal noiseless device was obtained by simulations and is also shown in the figure for reference. The left panel shows the results for one oracle call, and the right panel shows the results for two oracle calls. See also Table~\ref{tab:data-santiago-vs-simulation} for the numerical data used to generate the plots. The searched element was successfully marked with either set of gates, even though the multi-controlled $Z$ gates improved the results as expected. The simulated results confirm that we should apply two oracle calls in an ideal noiseless set-up. However, the experiments on real quantum devices clearly show that the second oracle call introduces noise and makes the results worse. Therefore, when executing Grover's algorithm on noisy quantum devices we must stop before the full number of iterations prescribed by theory, even though the success probability is not yet as high as we could expect from error-corrected devices.

\begin{figure}[h]
    \centering
    \includegraphics[width=0.49\textwidth]{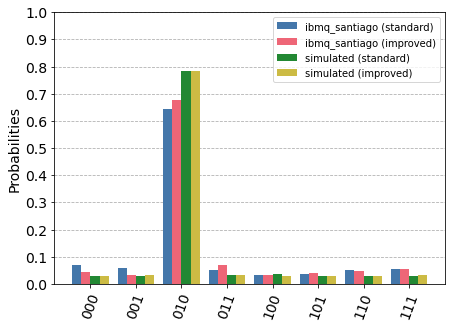}~
    \includegraphics[width=0.49\textwidth]{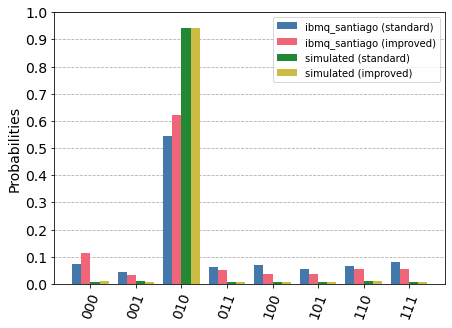}
    \caption{Results of the search algorithm on IBM Santiago using the standard set of gates and the improved set of gates. The searched element is $\ket{010}$. The results for an ideal noiseless device was obtained by simulations. (Left) One oracle call. (Right) Two oracle calls.}
    \label{fig:santiago-vs-simulation}
\end{figure}

\begin{table}[]
\centering
\begin{tabular}{@{}lllllllll@{}}
\toprule
Ket         & \multicolumn{4}{l}{One call}                                     & \multicolumn{4}{l}{Two calls}                                    \\ \midrule
            & \multicolumn{2}{l}{IBM Santiago} & \multicolumn{2}{l}{Simulator} & \multicolumn{2}{l}{IBM Santiago} & \multicolumn{2}{l}{Simulator} \\
            & std        & impr       & std      & impr      & std        & impr       & std      & impr      \\
$\ket{000}$ & 0.070           & 0.043          & 0.030         & 0.029         & 0.075           & 0.115          & 0.007         & 0.010         \\
$\ket{001}$ & 0.058           & 0.031          & 0.030         & 0.031         & 0.045           & 0.035          & 0.010         & 0.008         \\
$\ket{010}$ & 0.644           & 0.679          & 0.783         & 0.785         & 0.545           & 0.620          & 0.941         & 0.943         \\
$\ket{011}$ & 0.053           & 0.069          & 0.032         & 0.031         & 0.061           & 0.050          & 0.008         & 0.008         \\
$\ket{100}$ & 0.032           & 0.032          & 0.037         & 0.031         & 0.069           & 0.035          & 0.008         & 0.007         \\
$\ket{101}$ & 0.036           & 0.041          & 0.030         & 0.030         & 0.055           & 0.036          & 0.008         & 0.008         \\
$\ket{110}$ & 0.053           & 0.049          & 0.030         & 0.030         & 0.067           & 0.054          & 0.009         & 0.009         \\
$\ket{111}$ & 0.054           & 0.056          & 0.028         & 0.032         & 0.082           & 0.054          & 0.009         & 0.007         \\ \bottomrule
\end{tabular}
\caption{Numerical data used to generate Fig.~\ref{fig:santiago-vs-simulation}.}
\label{tab:data-santiago-vs-simulation}
\end{table}

In Fig.~\ref{fig:all-devices}, we have a summary of the results obtained from the search algorithm running on IBM Lagos, IBM Santiago, and IBM Bogota when the searched element was $\ket{010}$. The horizontal axis depicts the possible outcomes, and the vertical axis depicts the probability of each outcome based on 8196 independent runs performed on the device. We also compare the results obtained using the standard set of gates and the improved decomposition of multi-controlled $Z$ gates. The left panel shows the results for one oracle call, and the right panel shows the results for two oracle calls. See also Tables~\ref{tab:data-all-devices-1} and \ref{tab:data-all-devices-2} for the numerical data used to generate the plots---the former corresponding to one oracle call, and the latter corresponding to two oracle calls. Notice that all quantum devices were successful in finding the searched element, even though some of them performed much better than others. When choosing a NISQ computer \cite{Pre18} to solve a real-world problem, it is important to analyse each device in depth, since the quality of the results strongly depends on that choice. It is also important to notice that the second oracle call decreased the success probability on all quantum devices. When programming a NISQ computer, one must keep in mind that adaptations on the algorithms are often necessary in order to reduce the size and depth of the final circuit and compensate the effects of noise and decoherence.

\begin{figure}[h]
    \centering
    \includegraphics[width=0.49\textwidth]{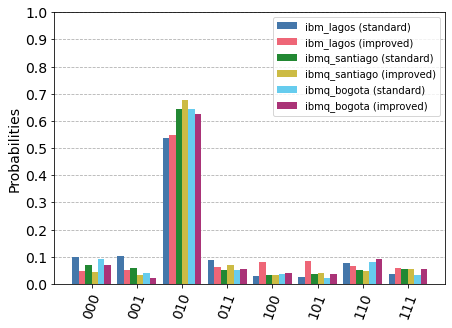}~
    \includegraphics[width=0.49\textwidth]{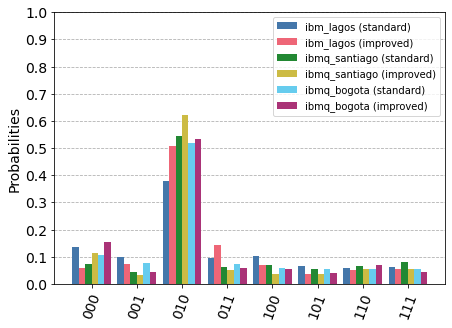}
    \caption{Results of the search algorithm on IBM Lagos, IBM Santiago and IBM Bogota using the standard set of gates and the improved set of gates. The searched element is $\ket{010}$. (Left) One oracle call. (Right) Two oracle calls.}
    \label{fig:all-devices}
\end{figure}

\begin{table}[]
\centering
\begin{tabular}{@{}lllllll@{}}
\toprule
Ket         & \multicolumn{2}{l}{IBM Lagos} & \multicolumn{2}{l}{IBM Santiago} & \multicolumn{2}{l}{IBM Bogota} \\ \midrule
            & std           & impr          & std             & impr           & std            & impr          \\
$\ket{000}$ & 0.098         & 0.048         & 0.070           & 0.043          & 0.090          & 0.069         \\
$\ket{001}$ & 0.105         & 0.051         & 0.058           & 0.031          & 0.042          & 0.024         \\
$\ket{010}$ & 0.538         & 0.547         & 0.644           & 0.679          & 0.643          & 0.627         \\
$\ket{011}$ & 0.089         & 0.062         & 0.053           & 0.069          & 0.051          & 0.057         \\
$\ket{100}$ & 0.030         & 0.082         & 0.032           & 0.032          & 0.038          & 0.039         \\
$\ket{101}$ & 0.026         & 0.084         & 0.036           & 0.041          & 0.023          & 0.037         \\
$\ket{110}$ & 0.078         & 0.066         & 0.053           & 0.049          & 0.080          & 0.093         \\
$\ket{111}$ & 0.036         & 0.060         & 0.054           & 0.056          & 0.033          & 0.054         \\ \bottomrule
\end{tabular}
\caption{Numerical data used to generate the left panel of Fig.~\ref{fig:all-devices}, corresponding to one oracle call.}
\label{tab:data-all-devices-1}
\end{table}

\begin{table}[]
\centering
\begin{tabular}{@{}lllllll@{}}
\toprule
Ket         & \multicolumn{2}{l}{IBM Lagos} & \multicolumn{2}{l}{IBM Santiago} & \multicolumn{2}{l}{IBM Bogota} \\ \midrule
            & std           & impr          & std             & impr           & std            & impr          \\
$\ket{000}$ & 0.136         & 0.058         & 0.075           & 0.115          & 0.106          & 0.154         \\
$\ket{001}$ & 0.099         & 0.074         & 0.045           & 0.035          & 0.076          & 0.045         \\
$\ket{010}$ & 0.377         & 0.506         & 0.545           & 0.620          & 0.517          & 0.533         \\
$\ket{011}$ & 0.095         & 0.145         & 0.061           & 0.050          & 0.075          & 0.058         \\
$\ket{100}$ & 0.103         & 0.071         & 0.069           & 0.035          & 0.059          & 0.054         \\
$\ket{101}$ & 0.068         & 0.038         & 0.055           & 0.036          & 0.055          & 0.041         \\
$\ket{110}$ & 0.060         & 0.053         & 0.067           & 0.054          & 0.057          & 0.071         \\
$\ket{111}$ & 0.063         & 0.054         & 0.082           & 0.054          & 0.055          & 0.044         \\ \bottomrule
\end{tabular}
\caption{Numerical data used to generate the right panel of Fig.~\ref{fig:all-devices}, corresponding to two oracle calls.}
\label{tab:data-all-devices-2}
\end{table}

In order to assess the quality of the results it is important to measure the fidelity between the actual probability distribution obtained by the real quantum device and the ideal distribution obtained by the simulator.  If $\{p_x\}$ and $\{q_x\}$ are two probability distributions, then the total variation distance $d$ is given by~\cite{LevinPeresWilmer2006}
\[d = \frac{1}{2}\sum_x \abs*{p_x - q_x},\]
and the Hellinger distance $h$ is given by~\cite{LevinPeresWilmer2006}
\[h^2 = \frac{1}{2}\sum_x \left(\sqrt{p_x} - \sqrt{q_x}\right)^2.  \]
In this work, we used two different definitions of fidelity, both $1-d$ and $1-h$.

It is also useful to consider the selectivity, which is a parameter introduced by Wang et al.~\cite{WangKrstic} to measure how strong is the result of a quantum computation when compared to the second most probable outcome. In this work, we adopt the simplified definition from Zhang et al.~\cite{Zhang2021Korepin}, given by the ratio
\[ S = \frac{P_s}{P_{ns}}, \]
where $P_s$ is the success probability, i.e., the probability of measuring the searched element, and $P_{ns}$ is the maximum probability of measuring a non-searched element.

In Table~\ref{tab:one-call}, we summarize the quality of the results obtained with one oracle call on the quantum computers IBM Lagos, IBM Santiago, and IBM Bogota. The parameters used to assess the quality of the results were the fidelity, the selectivity, and the success probability. The values on the tables are averages taken from all experiments. This analysis is important because the quality of the results vary according to the marked element. Notice that the improved decomposition of multi-controlled $Z$ gates really helped us to achieve better results.

\begin{table}[h]
\begin{tabular}{@{}lllllll@{}}
\toprule
                    & \multicolumn{2}{l}{Lagos} & \multicolumn{2}{l}{Santiago} & \multicolumn{2}{l}{Bogota} \\ \midrule
                    & standard    & improved    & standard      & improved     & standard     & improved    \\
Fidelity ($1-d$)    & 0.789       & 0.810       & 0.815         & 0.815        & 0.804        & 0.796       \\
Fidelity ($1-h$)    & 0.825       & 0.838       & 0.849         & 0.849        & 0.831        & 0.828       \\
Selectivity         & 6.432       & 6.776       & 7.599         & 7.386        & 5.912        & 6.359       \\
Success probability & 0.578       & 0.599       & 0.598         & 0.599        & 0.593        & 0.583       \\ \bottomrule
\end{tabular}
\caption{Quality of the results obtained with one oracle call assessed according to different metrics.}
\label{tab:one-call}
\end{table}

In Table~\ref{tab:two-calls}, we summarize the quality of the results obtained with two oracle calls on the quantum computers IBM Lagos, IBM Santiago, and IBM Bogota. The parameters used to assess the quality of the results were the fidelity, the selectivity, and the success probability. Notice that applying two oracle calls is not a good strategy on noisy devices.

\begin{table}[h]
\begin{tabular}{@{}lllllll@{}}
\toprule
                    & \multicolumn{2}{l}{Lagos} & \multicolumn{2}{l}{Santiago} & \multicolumn{2}{l}{Bogota} \\ \midrule
                    & standard    & improved    & standard      & improved     & standard     & improved    \\
Fidelity ($1-d$)    & 0.527       & 0.553       & 0.580         & 0.615        & 0.548        & 0.601       \\
Fidelity ($1-h$)    & 0.593       & 0.610       & 0.632         & 0.654        & 0.609        & 0.644       \\
Selectivity         & 3.952       & 4.197       & 5.644         & 5.671        & 4.693        & 5.646       \\
Success probability & 0.473       & 0.498       & 0.526         & 0.560        & 0.493        & 0.546       \\ \bottomrule
\end{tabular}
\caption{Quality of the results obtained with two oracle calls assessed according to different metrics.}
\label{tab:two-calls}
\end{table}

\section{Final comments}

We have described implementations of a quantum algorithm that solves small instances of the Discretizable Molecular Distance Geometry Problem (DMDGP). The strategy is to use Grover's algorithm to search the domain of an oracle function for a point that satisfies all the information required to solve the DMDGP. It is not trivial to build the oracle and this is one of the main contribution of this work. The implementation of our algorithm on small-scale quantum computers shows the consistency of the procedure, which is more efficient than classical methods that employ brute force, because the number of steps of the Grover algorithm is quadratically smaller. We have shown probability distributions of actual implementations on IBM quantum computers and have compared with correct results simulated on classical computers. We remark that it is better take fewer steps than the number prescribed by Grover's algorithm because the accumulation of errors would degrade the quality of the results.

Due to uncertainties in NMR data, the natural extension of this work is the
development of a quantum approach to the DMDGP with interval distances,
which is one of the important open problems in Distance Geometry~\cite{liberti18}. There is an active research on this topic, like the interval BP
algorithm~\cite{dambrosio17,goncalves17,lavor13,souza11,souza13,worley18}
and other approaches based on Clifford algebra~\cite{alves10,lavor18,lavor19}.
\section*{Acknowledgments}

We acknowledge financial support of the Brazilian research agencies CNPq, FAPESP, and FAPERJ.


\begin{thebibliography}{10}

\bibitem{donald11}
B.~Donald.
\newblock {\em Algorithms in Structural Molecular Biology}.
\newblock MIT Press, 2011.

\bibitem{wuthrich89}
K.~W\"{u}thrich.
\newblock Protein structure determination in solution by nuclear magnetic
  resonance spectroscopy.
\newblock {\em Science}, 243:45--50, 1989.

\bibitem{crippen88}
G.~Crippen and T.~Havel.
\newblock {\em Distance Geometry and Molecular Conformation}.
\newblock Wiley, 1988.

\bibitem{billinge16}
S.~Billinge, P.~Duxbury, D.~Gon\c{c}alves, C.~Lavor, and A.~Mucherino.
\newblock Assigned and unassigned distance geometry: applications to biological
  molecules and nanostructures.
\newblock {\em 4OR}, 14:337--376, 2016.

\bibitem{billinge18}
S.~Billinge, P.~Duxbury, D.~Gon\c{c}alves, C.~Lavor, and A.~Mucherino.
\newblock Recent results on assigned and unassigned distance geometry with
  applications to protein molecules and nanostructures.
\newblock {\em Annals of Operations Research}, 271:161--203, 2018.

\bibitem{lavor17}
C.~Lavor, L.~Liberti, W.~Lodwick, and T.~Mendon\c{c}a da~Costa.
\newblock {\em An Introduction to Distance Geometry applied to Molecular
  Geometry}.
\newblock Springer, 2017.

\bibitem{liberti14a}
L.~Liberti, C.~Lavor, N.~Maculan, and A.~Mucherino.
\newblock Euclidean distance geometry and applications.
\newblock {\em SIAM Review}, 56:3--69, 2014.

\bibitem{liberti16}
L.~Liberti and C.~Lavor.
\newblock Six mathematical gems from the history of distance geometry.
\newblock {\em International Transactions in Operational Research},
  23:897--920, 2016.

\bibitem{liberti17}
L.~Liberti and C.~Lavor.
\newblock {\em Euclidean Distance Geometry: An Introduction}.
\newblock Springer, 2017.

\bibitem{mucherino13}
A.~Mucherino, C.~Lavor, L.~Liberti, and N.~Maculan, editors.
\newblock {\em Distance Geometry: Theory, Methods, and Applications}.
\newblock Springer, 2013.

\bibitem{grover97}
L.~Grover.
\newblock Quantum mechanics helps in searching for a needle in a haystack.
\newblock {\em Physical Review Letters}, 79:325--328, 1997.

\bibitem{Arute_2019}
Frank Arute, Kunal Arya, Ryan Babbush, Dave Bacon, Joseph~C. Bardin, Rami
  Barends, Rupak Biswas, Sergio Boixo, Fernando G. S.~L. Brandao, David~A.
  Buell, Brian Burkett, Yu~Chen, Zijun Chen, Ben Chiaro, Roberto Collins,
  William Courtney, Andrew Dunsworth, Edward Farhi, Brooks Foxen, Austin
  Fowler, Craig Gidney, Marissa Giustina, Rob Graff, Keith Guerin, Steve
  Habegger, Matthew~P. Harrigan, Michael~J. Hartmann, Alan Ho, Markus Hoffmann,
  Trent Huang, Travis~S. Humble, Sergei~V. Isakov, Evan Jeffrey, Zhang Jiang,
  Dvir Kafri, Kostyantyn Kechedzhi, Julian Kelly, Paul~V. Klimov, Sergey Knysh,
  Alexander Korotkov, Fedor Kostritsa, David Landhuis, Mike Lindmark, Erik
  Lucero, Dmitry Lyakh, Salvatore Mandr{\`a}, Jarrod~R. McClean, Matthew
  McEwen, Anthony Megrant, Xiao Mi, Kristel Michielsen, Masoud Mohseni, Josh
  Mutus, Ofer Naaman, Matthew Neeley, Charles Neill, Murphy~Yuezhen Niu, Eric
  Ostby, Andre Petukhov, John~C. Platt, Chris Quintana, Eleanor~G. Rieffel,
  Pedram Roushan, Nicholas~C. Rubin, Daniel Sank, Kevin~J. Satzinger, Vadim
  Smelyanskiy, Kevin~J. Sung, Matthew~D. Trevithick, Amit Vainsencher, Benjamin
  Villalonga, Theodore White, Z.~Jamie Yao, Ping Yeh, Adam Zalcman, Hartmut
  Neven, and John~M. Martinis.
\newblock Quantum supremacy using a programmable superconducting processor.
\newblock {\em Nature}, 574(7779):505--510, Oct 2019.

\bibitem{Yulin_Wu_2021}
Yulin Wu, Wan-Su Bao, Sirui Cao, Fusheng Chen, Ming-Cheng Chen, Xiawei Chen,
  Tung-Hsun Chung, Hui Deng, Yajie Du, Daojin Fan, Ming Gong, Cheng Guo, Chu
  Guo, Shaojun Guo, Lianchen Han, Linyin Hong, He-Liang Huang, Yong-Heng Huo,
  Liping Li, Na~Li, Shaowei Li, Yuan Li, Futian Liang, Chun Lin, Jin Lin,
  Haoran Qian, Dan Qiao, Hao Rong, Hong Su, Lihua Sun, Liangyuan Wang, Shiyu
  Wang, Dachao Wu, Yu~Xu, Kai Yan, Weifeng Yang, Yang Yang, Yangsen Ye,
  Jianghan Yin, Chong Ying, Jiale Yu, Chen Zha, Cha Zhang, Haibin Zhang, Kaili
  Zhang, Yiming Zhang, Han Zhao, Youwei Zhao, Liang Zhou, Qingling Zhu,
  Chao-Yang Lu, Cheng-Zhi Peng, Xiaobo Zhu, and Jian-Wei Pan.
\newblock Strong quantum computational advantage using a superconducting
  quantum processor.
\newblock {\em Phys. Rev. Lett.}, 127:180501, Oct 2021.

\bibitem{Pre18}
John Preskill.
\newblock Quantum {C}omputing in the {NISQ} era and beyond.
\newblock {\em {Quantum}}, 2:79, August 2018.

\bibitem{lavor05}
C.~Lavor, L.~Liberti, and N.~Maculan.
\newblock Grover's algorithm applied to the molecular distance geometry
  problem.
\newblock In {\em Proceedings of the 7th Brazilian Congress of Neural
  Networks}, 2005.

\bibitem{sylvester77}
J.~Sylvester.
\newblock Chemistry and algebra.
\newblock {\em Nature}, 17:284--284, 1877.

\bibitem{bajaj88}
C.~Bajaj.
\newblock The algebraic degree of geometric optimization problems.
\newblock {\em Discrete and Computational Geometry}, 3:177--191, 1988.

\bibitem{lara14}
P.~Lara, R.~Portugal, and C.~Lavor.
\newblock A new hybrid classical-quantum algorithm for continuous global
  optimization problems.
\newblock {\em Journal of Global Optimization}, 60:317--331, 2014.

\bibitem{lavor06}
C.~Lavor, L.~Liberti, and N.~Maculan.
\newblock Computational experience with the molecular distance geometry
  problem.
\newblock In J.~Pint\'{e}r, editor, {\em Global Optimization: Scientific and
  Engineering Case Studies}, pages 213--225. Springer, 2006.

\bibitem{liberti10}
L.~Liberti, C.~Lavor, A.~Mucherino, and N.~Maculan.
\newblock Molecular distance geometry methods: from continuous to discrete.
\newblock {\em International Transactions in Operational Research}, 18:33--51,
  2010.

\bibitem{lavor12b}
C.~Lavor, L.~Liberti, N.~Maculan, and A.~Mucherino.
\newblock The discretizable molecular distance geometry problem.
\newblock {\em Computational Optimization and Applications}, 52:115--146, 2012.

\bibitem{cassioli15}
A.~Cassioli, B.~Bordiaux, G.~Bouvier, A.~Mucherino, R.~Alves, L.~Liberti,
  M.~Nilges, C.~Lavor, and T.~Malliavin.
\newblock An algorithm to enumerate all possible protein conformations
  verifying a set of distance constraints.
\newblock {\em BMC Bioinformatics}, 16:16--23, 2015.

\bibitem{lavor12c}
C.~Lavor, L.~Liberti, N.~Maculan, and A.~Mucherino.
\newblock Recent advances on the discretizable molecular distance geometry
  problem.
\newblock {\em European Journal of Operational Research}, 219:698--706, 2012.

\bibitem{malliavin19}
T.~Malliavin, A.~Mucherino, C.~Lavor, and L.~Liberti.
\newblock Systematic exploration of protein conformational space using a
  distance geometry approach.
\newblock {\em Journal of Chemical Information and Modeling}, 59:4486--4503,
  2019.

\bibitem{cassioli13}
A.~Cassioli, O.~Gunluk, C.~Lavor, and L.~Liberti.
\newblock Discretization vertex orders in distance geometry.
\newblock {\em Discrete Applied Mathematics}, 197:27--41, 2015.

\bibitem{lavor12a}
C.~Lavor, J.~Lee, A.~Lee-St John, L.~Liberti, A.~Mucherino, and M.~Sviridenko.
\newblock Discretization orders for distance geometry problems.
\newblock {\em Optimization Letters}, 6:783--796, 2012.

\bibitem{lavor19a}
C.~Lavor, L.~Liberti, B.~Donald, B.~Worley, B.~Bardiaux, T.~Malliavin, and
  M.~Nilges.
\newblock Minimal nmr distance information for rigidity of protein graphs.
\newblock {\em Discrete Applied Mathematics}, 256:91--104, 2019.

\bibitem{lavor19b}
C.~Lavor, M.~Souza, L.~Mariano, and L.~Liberti.
\newblock On the polinomiality of finding $^{K}${DMDGP} re-orders.
\newblock {\em Discrete Applied Mathematics}, 267:190--194, 2019.

\bibitem{carvalho08}
R.~Carvalho, C.~Lavor, and F.~Protti.
\newblock Extending the geometric build-up algorithm for the molecular distance
  geometry problem.
\newblock {\em Information Processing Letters}, 108:234--237, 2008.

\bibitem{liberti08}
L.~Liberti, C.~Lavor, and N.~Maculan.
\newblock A branch-and-prune algorithm for the molecular distance geometry
  problem.
\newblock {\em International Transactions in Operational Research}, 15:1--17,
  2008.

\bibitem{gonalves21}
D.~Gon{\c{c}}alves, C.~Lavor, L.~Liberti, and M.~Souza.
\newblock A new algorithm for the $^k${DMDGP} subclass of distance geometry
  problems with exact distances.
\newblock {\em Algorithmica}, 83:2400--2426, 2021.

\bibitem{lavor21}
C.~Lavor, A.~Oliveira, W.~Rocha, and M.~Souza.
\newblock On the optimality of finding {DMDGP} symmetries.
\newblock {\em Computational and Applied Mathematics}, 40:98, 2021.

\bibitem{liberti14b}
L.~Liberti, B.~Masson, J.~Lee, C.~Lavor, and A.~Mucherino.
\newblock On the number of realizations of certain henneberg graphs arising in
  protein conformation.
\newblock {\em Discrete Applied Mathematics}, 165:213--232, 2014.

\bibitem{mucherino12}
A.~Mucherino, C.~Lavor, and L.~Liberti.
\newblock Exploiting symmetry properties of the discretizable molecular
  distance geometry problem.
\newblock {\em Journal of Bioinformatics and Computational Biology},
  10(12420):09, 2012.

\bibitem{liberti13a}
L.~Liberti, C.~Lavor, J.~Alencar, and G.~Resende.
\newblock Counting the number of solutions of $^{K}${DMDGP} instances.
\newblock {\em Lecture Notes in Computer Science}, 8085:224--230, 2013.

\bibitem{lavor07}
C.~Lavor, L.~Carvalho, R.~Portugal, and C.~Moura.
\newblock Complexity of {G}rover's algorithm: an algebraic approach.
\newblock {\em International Journal of Applied Mathematics}, 20:801--814,
  2007.

\bibitem{marquezino19}
F.~Marquezino, R.~Portugal, and C.~Lavor.
\newblock {\em A Primer on Quantum Computing}.
\newblock Springer, 2019.

\bibitem{portugal18}
R.~Portugal.
\newblock {\em Quantum Walks and Search Algorithms}.
\newblock Springer, 2nd edition, 2018.

\bibitem{nielsen00}
M.~Nielsen and I.~Chuang.
\newblock {\em Quantum Computation and Quantum Information}.
\newblock Cambridge University Press, 2000.

\bibitem{lavor15}
C.~Lavor, R.~Alves, W.~Figueiredo, A.~Petraglia, and N.~Maculan.
\newblock Clifford algebra and the discretizable molecular distance geometry
  problem.
\newblock {\em Advances in Applied Clifford Algebra}, 25:925--942, 2015.

\bibitem{LevinPeresWilmer2006}
David~A. Levin, Yuval Peres, and Elizabeth~L. Wilmer.
\newblock {\em {Markov chains and mixing times}}.
\newblock American Mathematical Society, 2006.

\bibitem{WangKrstic}
Yulun Wang and Predrag~S. Krstic.
\newblock Prospect of using {G}rover's search in the noisy-intermediate-scale
  quantum-computer era.
\newblock {\em Phys. Rev. A}, 102:042609, Oct 2020.

\bibitem{Zhang2021Korepin}
Kun Zhang, Pooja Rao, Kwangmin Yu, Hyunkyung Lim, and Vladimir Korepin.
\newblock Implementation of efficient quantum search algorithms on nisq
  computers.
\newblock {\em Quantum Information Processing}, 20(7):233, Jul 2021.

\bibitem{liberti18}
L.~Liberti and C.~Lavor.
\newblock Open research areas in distance geometry.
\newblock In A.~Migalas and P.~Pardalos, editors, {\em Open Problems in
  Optimization and Data Analysis}, pages 183--223. Springer, 2018.

\bibitem{dambrosio17}
C.~Dambrosio, V.~Ky, C.~Lavor, L.~Liberti, and N.~Maculan.
\newblock New error measures and methods for realizing protein graphs from
  distance data.
\newblock {\em Discrete and Computational Geometry}, 57:371--418, 2017.

\bibitem{goncalves17}
D.~Gon\c{c}alves, A.~Mucherino, C.~Lavor, and L.~Liberti.
\newblock Recent advances on the interval distance geometry problem.
\newblock {\em Journal of Global Optimization}, 69:525--545, 2017.

\bibitem{lavor13}
C.~Lavor, L.~Liberti, and A.~Mucherino.
\newblock The interval branch-and-prune algorithm for the discretizable
  molecular distance geometry problem with inexact distances.
\newblock {\em Journal of Global Optimization}, 56:855--871, 2013.

\bibitem{souza11}
M.~Souza, A.~Xavier, C.~Lavor, and N.~Maculan.
\newblock Hyperbolic smoothing and penalty techniques applied to molecular
  structure determination.
\newblock {\em Operations Research Letters}, 39:461--465, 2011.

\bibitem{souza13}
M.~Souza, C.~Lavor, A.~Muritiba, and N.~Maculan.
\newblock Solving the molecular distance geometry problem with inaccurate
  distance data.
\newblock {\em BMC Bioinformatics}, 14:S71--S76, 2013.

\bibitem{worley18}
B.~Worley, F.~Delhommel, F.~Cordier, T.~Malliavin, B.~Bardiaux, N.~Wolff,
  M.~Nilges, C.~Lavor, and L.~Liberti.
\newblock Tuning interval branch-and-prune for protein structure determination.
\newblock {\em Journal of Global Optimization}, 72:109--127, 2018.

\bibitem{alves10}
R.~Alves and C.~Lavor.
\newblock {C}lifford algebra applied to {G}rover's algorithm.
\newblock {\em Advances in Applied Clifford Algebras}, 20:477--488, 2010.

\bibitem{lavor18}
C.~Lavor, S.~Xamb\'{o}-Descamps, and I.~Zaplana.
\newblock {\em A Geometric Algebra Invitation to Space-Time Physics Robotics
  and Molecular Geometry}.
\newblock Springer, 2018.

\bibitem{lavor19}
C.~Lavor and R.~Alves.
\newblock Oriented conformal geometric algebra and the molecular distance
  geometry problem.
\newblock {\em Advances in Applied Clifford Algebra}, 29:1--19, 2019.

\end{thebibliography}
\end{document}